\documentclass[sigconf]{acmart}

\usepackage{amsmath,amsfonts,amsthm}
\usepackage{hyperref}
\usepackage{subcaption}
\usepackage{algorithm}
\usepackage{algorithmic}
\usepackage{enumitem}
\usepackage{thmtools}
\usepackage{thm-restate}
\usepackage{multirow}
\usepackage[utf8]{inputenc}

\newtheorem*{theorem*}{Theorem}

\newcommand{\mk}[1]{\textcolor{red}{\textbf{Mehrdad:} #1}}

\AtBeginDocument{%
  }

\setcopyright{acmcopyright}
\copyrightyear{2024}
\acmYear{2024}
\acmDOI{XXXXXXX.XXXXXXX}

\acmConference[PODC 2024]{ACM Symposium on Principles of Distributed Computing}{June 17--21,
  2024}{Nantes, France}

\acmISBN{978-1-4503-XXXX-X/18/06}

\newcommand{\MDP}{\mathcal{M}}
\newcommand{\States}{S}
\newcommand{\Actions}{A}

\newcommand{\transition}{P}

\newcommand{\Config}{\mathcal{C}}
\newcommand{\TREE}{\mathcal{T}}
\newcommand{\honest}{\mathcal{H}}
\newcommand{\Own}{\mathcal{O}}
\newcommand{\type}{\textsc{type}}
\newcommand{\MCAttack}{\textsc{tree attack}}
\newcommand{\RER}{\texttt{ERRev}}
\newcommand{\MP}{\texttt{MP}}

\begin{document}

\title{Fully Automated Selfish Mining Analysis in Efficient Proof Systems Blockchains}


\author{Krishnendu Chatterjee}
\orcid{0000-0002-4561-241X}
\email{krishnendu.chatterjee@ist.ac.at}
\affiliation{%
  \institution{IST Austria}
  \country{Austria}
}

\author{Amirali Ebrahimzadeh}
\orcid{0000-0002-9887-4341}
\email{ebrahimzadeh.amirali@gmail.com}
\affiliation{%
  \institution{Sharif University of Technology}
  \country{Iran}
}

\author{Mehrdad Karrabi}
\orcid{0009-0007-5253-9170}
\email{mehrdad.karrabi@ist.ac.at}
\affiliation{%
  \institution{IST Austria}
  \country{Austria}
}

\author{Krzysztof Pietrzak}
\orcid{0000-0002-9139-1654}
\email{krzysztof.pietrzak@ist.ac.at}
\affiliation{%
  \institution{IST Austria}
  \country{Austria}
}

\author{Michelle Yeo}
\orcid{0009-0001-3676-4809}
\email{mxyeo@nus.edu.sg}
\affiliation{%
  \institution{National University of Singapore}
  \country{Singapore}
}

\author{\DJ or\dj e \v{Z}ikeli\'c}
\orcid{0000-0002-4681-1699}
\email{dzikelic@smu.edu.sg}
\affiliation{%
  \institution{Singapore Management University}
  \country{Singapore}
}

\renewcommand{\shortauthors}{K. Chatterjee et al.}

\begin{abstract}
We study selfish mining attacks in longest-chain blockchains like Bitcoin, but where the proof of work is replaced with efficient proof systems -- like proofs of stake or proofs of space -- and consider the problem of computing an optimal selfish mining attack which maximizes expected relative revenue of the adversary, thus minimizing the chain quality. To this end, we propose a novel selfish mining attack that aims to maximize this objective and formally model the attack as a Markov decision process (MDP). We then present a formal analysis procedure which computes an $\epsilon$-tight lower bound on the optimal expected relative revenue in the MDP and a strategy that achieves this $\epsilon$-tight lower bound, where $\epsilon>0$ may be any specified precision. Our analysis is fully automated and provides formal guarantees on the correctness. We evaluate our selfish mining attack and observe that it achieves superior expected relative revenue compared to two considered baselines.

In concurrent work [Sarenche FC'24] does an automated analysis on selfish mining in \emph{predictable} longest-chain blockchains based on efficient proof systems. Predictable means the randomness for the challenges is fixed for many blocks (as used e.g., in Ouroboros), while 
we consider \emph{unpredictable} (Bitcoin-like) chains where the challenge is derived from the previous block.

\end{abstract}




\copyrightyear{2024}
\acmYear{2024}
\setcopyright{rightsretained}
\acmConference[PODC '24]{ACM Symposium on Principles of Distributed
Computing}{June 17--21, 2024}{Nantes, France}
\acmBooktitle{ACM Symposium on Principles of Distributed Computing (PODC
'24), June 17--21, 2024, Nantes, France}\acmDOI{10.1145/3662158.3662769}
\acmISBN{979-8-4007-0668-4/24/06}

\maketitle

\section{Introduction}\label{sec:intro}

\paragraph{Bitcoin.}
Blockchain protocols were proposed as a solution to achieve consensus over some states (e.g., financial transactions) in a distributed and permissionless (i.e., everyone can participate in securing the chain) setting. Participants in blockchain protocols add data into blocks, which are then appended to the blockchain with some probability that depends on the underlying consensus protocol.
The earliest and most commonly adopted consensus protocol is proof of work (PoW), which also forms the basis of the Bitcoin blockchain protocol~\cite{nakamoto2008bitcoin}. 

\smallskip{\em Proof of work.}
The parties that maintain a PoW blockchain like Bitcoin are called miners. The general idea is that in order to add a block to the chain, the miners derive a computationally hard but easily verifiable puzzle from the tip of the chain. To add a block to the chain, this block must contain a solution to this puzzle. This mechanism ensures that attacking the chain, in particular rewriting past blocks in a double spending attack, is computationally very expensive. 
In Bitcoin, the puzzle is a hashcash style PoW~\cite{hashcash}, with parameters being a difficulty level $D$ and a global hash function (e.g., SHA256). 
A newly generated block can only be added to the chain if the hash of the new block, the previous block, the miner's public key and some nonce give a value that is less than the current difficulty $D$.
To mine a block in Bitcoin, miners will continuously generate different nonces and hash them until they find a nonce that passes the threshold. 
The lucky miner that first finds a block broadcasts it across the network, where other miners can easily verify the validity of the nonce. 
The Bitcoin protocol specifies that miners should always work towards extending the longest chain they are aware of, hence such chains are called ``longest-chain blockchains''. An alternative approach is chains based on Byzantine Fault Tolerance like  Algorand~\cite{algorand} where randomly rotating committees approve blocks to be added.


\smallskip{\em Efficient proof systems.}
As an alternative to PoW, several other consensus protocols based on efficient proof systems  have been proposed~\cite{ouroboros,spacemint,CohenP23}. 
Generally, we say a ``proof of X'' is efficient if computing a proof for a random challenge is very cheap assuming one has the resource X. 
The most popular and best investigated efficient proofs are Proofs of Stake (PoStake) as e.g., used in Ouroboros~\cite{ouroboros} or post-merge Ethereum. Here the coins recorded in the blockchain are the resource. Another proposal of an actual physical resource are Proofs of Space (PoSpace)~\cite{DziembowskiFKP15}, where the resource is disk-space. The first proposal for a PoSpace based chain was Spacemint~\cite{spacemint}. The first deployed chain was the Chia network~\cite{CohenP23}, it uses PoSpace in combination with verifiable delay functions (VDFs) to address some security challenges, and thus is referred to as a Proofs of Space and Time (PoST) based chain. 
See Appendix~\ref{sec:efficientproofsystems} for more details of blockchains based on efficient proof systems that we consider in our work.


\smallskip{\em Selfish mining.}
There are various security properties we want from longest-chain blockchains, the most important ones being persistence and liveness~\cite{GarayKL15}. Informally, persistence means that entries added to the chain will remain there forever, while liveness means that the chain remains available. 
A less obvious requirement is fairness, in the sense that a party that contributes a $p$ fraction of the resource (hashing power in PoW, space in PoST, staked coins in PoStake) should contribute a $p$ fraction of the blocks, and thus get a $p$ fraction of the rewards. 

It was first observed in~\cite{EyalS18} that Bitcoin is not fair in this sense due to \emph{selfish mining} attacks. In selfish mining, attackers mine blocks but occasionally selectively withhold these (so the honest miners cannot mine on top of those) to later release them, overtaking the honest chain and thus orphaning honest blocks. In doing so, the attackers can reduce the \emph{chain quality} of the blockchain~\cite{EyalS18}, this measure quantifies the fraction of blocks contributed by an adversarial miner. 

Although the analysis of selfish mining and its impact on chain quality is well-studied in Bitcoin and other PoW-based blockchains \cite{EyalS18,SapirshteinSZ16,ZurET20,squirl}, a thorough analysis of optimal selfish mining attacks in blockchains based on efficient proof systems is difficult due the fact that the computational cost of generating proofs is low in these systems. This gives rise to several issues that in the PoStake setting are generally referred to as the ``nothing-at-stake'' (NaS) problem.
\footnote{We will slightly abuse terminology in our work and continue to use the terms ``mining'' and ``miners'' from PoW-based chains also when discussing chains based on efficient proof systems even though when using  PoStake this is sometimes referred to as ``proposing'' and ``proposers'' (but it is not coherent, some works like~\cite{FerreiraW21,FerreiraHWY22} use mining) while in PoSpace it was suggested to use the term ``farming'' and ``farmers''.}

\smallskip{\em The security vs predictability dilemma.}
Assume one would na\"ively replace PoW in Bitcoin with an efficient proof system like PoStake. 
As computing proofs is now cheap, an adversarial miner can try to extend blocks at different depths (not just on the tip of the longest chain), growing private trees at each depth. If they manage to create a tree of depth $d$ that starts less than $d$ blocks deep in the public chain, releasing the longest path in this tree will force the honest miners to switch to this path. 
Such ``tree growing'' attacks can be prevented by diverting from Bitcoin-like protocols, and instead of using the block at depth $i$ to derive the challenge for the block at depth $i+1$, one uses some fixed randomness for a certain number of consecutive blocks. 
Unfortunately, this creates a new security issue as an adversary can now predict when in the future they will be able to create blocks~\cite{CohenNPW18,BagariaDKOTVWZ22}.
We note that this approach is used in the PoStake-based Ouroboros~\cite{ouroboros} chain, where one only creates a fresh challenge every five days. 
An extreme in the other direction is the PoSpace-based Spacemint~\cite{spacemint} chain or an early proposal of the PoST-based Chia chain~\cite{CohenP19} which, like Bitcoin, derive the challenge from the previous block (the deployed Chia design~\cite{CohenP23} uses a fresh challenge every 10 minutes, or 32 blocks).

\smallskip{\em Limitations of previous analyses.}
The persistence, i.e., security against \emph{double spending} attacks, in predictable (Ouroboros-like) and unpredictable (Bitcoin-like) longest-chain blockchains is pretty well understood~\cite{BagariaDKOTVWZ22}. In particular, using the ``tree growing'' attack outlined above, an adversary can ``amplify'' its resources in an unpredictable protocol by a factor of $e\approx 2.72$. Thus, security requires that the adversary controls less than $\frac{1}{1+e}\approx 0.269$ fraction of the total resources. In a predictable protocol, it is sufficient for the adversary to have $< \frac{1}{2}$ of the resource. This is better, but being predictable opens new attack vectors like bribing attacks. 
Unfortunately, the security of blockchain protocols based on efficient proof systems in the face of \emph{selfish mining} attacks is a lot less understood. In particular, previous analyses in this area suffer mainly from two limitations:
\begin{enumerate}[leftmargin=*]
    \item (Model) 
    Although there have also been some very recent works analysing \emph{selfish mining} attacks in blockchains based on efficient proof systems~\cite{CohenP23,WangKBKOTV19,FanZ17,FerreiraW21,squirl,RoozbehSB24}, these analyses so far have only focused on predictable protocols. 
    \item (Methodology) Furthermore, they either consider different adversarial objectives~\cite{CohenP23,WangKBKOTV19,FanZ17,FerreiraW21} or use deep reinforcement learning to obtain selfish mining strategies~\cite{squirl,RoozbehSB24}. However, deep reinforcement learning only empirically maximizes the objective and does not provide \emph{formal guarantees} on the quality (lower or upper bound) of learned strategies.
\end{enumerate}




{\color{blue}

}

\smallskip{\em Our approach.}
In this work, we present the first analysis of selfish mining in unpredictable longest-chain blockchains based on efficient proof systems. Recall that in such chains the challenge for each block is determined by the previous block. 
At each time step, the adversary has to decide whether to mine new blocks or to publish one of their private forks which is longer than the public chain. 
Our analysis is concerned with finding the optimal sequence of mining and fork reveal actions that the adversary should follow in order to maximize the expected relative revenue (i.e.,~the ratio of the expected adversarial blocks in the main chain when following a given strategy compared to the total number of blocks in the chain). Note that this is a challenging problem. Since at each time step the adversary can choose between mining new blocks and publishing {\em any} of its private chains, the strategy space of the adversary is exponential in the number of privately mined forks which makes the manual formal analysis intractable.

To overcome this challenge, we model our selfish mining attack as a {\em finite-state Markov Decision Process} {\em (MDP)}~\cite{Puterman94}. We then present a formal analysis procedure which, given a precision parameter $\epsilon > 0$, computes 
\begin{itemize}[leftmargin=*]
    \item an {\em $\epsilon$-tight lower bound on the optimal expected relative revenue} that a selfish mining strategy in the MDP can achieve, and
    \item a {\em selfish mining strategy} achieving this $\epsilon$-tight lower bound.
\end{itemize}
At the core of our formal analysis procedure lies a reduction from the problem of computing an optimal selfish mining strategy under the expected relative revenue objective to computing an optimal strategy in the MDP under a {\em mean-payoff objective} for a suitably designed reward function. While we defer the details on MDPs and mean-payoff objectives to Section~\ref{sec:prelims}, we note that solving mean-payoff finite-state MDPs is a classic and well-studied problem within the formal methods community for which efficient (polynomial-time) algorithms exist~\cite{Puterman94,filar2012competitive} and there are well-developed tools that implement them together with further optimizations~\cite{storm,prism}. These algorithms are {\em fully automated} and provide {\em formal guarantees} on the correctness of their outputs, and the formal analysis of our selfish mining attack naturally inherits these desirable features.

\smallskip{\em Contributions.} Our contributions can be summarized as follows:
\begin{enumerate}[leftmargin=*]
    \item We study selfish mining in unpredictable efficient proof systems blockchain protocols where the adversary's goal is to maximize {\em expected relative revenue} and thus minimize chain quality.
    \item We propose a {\em novel selfish mining attack} that optimizes the expected relative revenue in unpredictable blockchain protocols based on efficient proof systems. We {\em formally model} the attack as an MDP.
    \item We present a {\em formal analysis procedure} for our selfish mining attack. Given a precision parameter $\epsilon > 0$, our formal analysis procedure computes an {\em $\epsilon$-tight} lower bound on the optimal expected relative revenue in the MDP together with a {\em selfish mining} strategy that achieves it. The procedure is {\em fully automated} and provides {\em formal guarantees} on its correctness.
    \item Our formal analysis is {\em flexible} to changes in system model parameter values. For instance, it allows us to tweak system model parameters and study their impact on the optimal expected relative revenue that selfish mining can achieve in a {\em fully automated fashion} and without the need to develop novel analyses for different parameter values. This is in stark contrast to formal analyses whose correctness is proved manually, see Section~\ref{sec:methdiscussion} for a more detailed discussion.
    \item We implement the MDP model and the formal analysis procedure and {\em experimentally evaluate} the quality of the expected relative revenue achieved by the computed selfish mining strategy. We compare our selfish mining attack to two baselines: (1)~honest mining strategy and (2)~a direct extension of the PoW selfish mining strategy of~\cite{EyalS18} to the setting of blockchains based on efficient proof systems. Our experiments show that our selfish mining strategy achieves higher expected relative revenue compared to the two baselines. 
\end{enumerate}

\smallskip{\em A Remark on the Model.}
While the selfish mining attacks analyzed in this paper apply to PoStake, PoSpace and PoST based longest-chain blockchains, they capture PoST better than PoStake and PoSpace as we will shortly outline now. Due to the use of VDFs, PoST is ``strongly unpredictable'' in the sense that a miner cannot predict when they will be able to extend any block, while PoStake and PoSpace are ``weakly unpredictable'' in the sense that a miner can predict when they will find blocks on top of their own (but not other) blocks. Moreover, in PoST a malicious miner must run a VDF on top of every block they try to extend, while in PoStake or PoSpace this comes basically for free. 
The class of selfish mining attacks analyzed in this paper does not exploit ``weak unpredictability'' (a necessary assumption for PoST, but not PoStake or PoSpace) and, to be able to give automated bounds, we also assume a bound on the number of blocks one tries to extend, which is a realistic assumption for PoST (as each block requires a VDF) but less so for PoSpace or PoStake.

\subsection{Related Work}\label{sec:related}
\paragraph{Selfish mining in Bitcoin.}
One of the motivations behind the initial design of Bitcoin and other PoW blockchain systems is fairness.
That is, a miner controlling $p \in [0,1]$ proportion of resources should generate blocks at the rate of $p$.
This led to the initial analysis claiming Bitcoin is fair when the total resource $p$ owned by adversarial miners is bounded from above by $\frac{1}{2}$~\cite{GarayKL15}.
Nevertheless,~\cite{EyalS18} outlined an attack called ``selfish mining" in Bitcoin which shows that even for $p < \frac{1}{2}$ it is possible to generate blocks at a rate strictly greater than $p$, implying that Bitcoin is inherently ``unfair".
The attack secretly forks the main chain and mines blocks in a private, unannounced chain. These blocks are only revealed when the private chain is longer than the main chain, thus forcing the honest miners to switch from the main chain to the private chain. 
This causes the honest miners to waste their computational resources on the now shorter public chain.
It is shown in~\cite{EyalS18} that due to selfish mining, the largest amount of adversarial resources that can be tolerated while ensuring the security of Bitcoin is $\frac{1}{3}$ in the optimistic setting where honest miners propagate their blocks first, and $\frac{1}{4}$ if both honest and adversarial miners propagate their blocks first with equal probability.

\smallskip{\em NaS selfish mining.}
NaS selfish mining can be studied under several different adversarial objectives, see Appendix~\ref{app:sm_objectives} for a thorough discussion and comparison of these objectives. The objective we consider in our work is the expected relative revenue of the adversary, which is also considered in the analysis of~\cite{RoozbehSB24} for PoStake.
There have been earlier analyses that consider several other objectives, however. The works of~\cite{WangKBKOTV19} and~\cite{CohenP23} studied the probability of an adversarial coalition overtaking the honest chain and showed respectively that under this objective, the largest fraction of adversarial resources that can be tolerated in PoStake based blockchains is $\frac{1}{1+e} \approx 0.269$, and between $\frac{1}{1+e} < p < \frac{1}{2}$ for Chia, a PoST-based blockchain.
The work of~\cite{FanZ17} studied selfish mining strategies in PoStake blockchains under the objective of finding strategies that give the adversarial coalition a larger payoff compared to following the honest strategy. They showed that the expected advantage of the adversary when growing a private tree from the genesis block is $e$ times larger than the expected revenue of following the honest strategy, implying that the maximum fraction of adversarial resource that can be tolerated to be secure under this strategy and objective is $\frac{1}{e} \approx 0.37$.

\smallskip{\em Optimal selfish mining strategies.}
There have been some works that go beyond proposing and analysing specific selfish mining strategies to actually claiming the \emph{optimality} of these strategies.
The work of~\cite{SapirshteinSZ16} modelled the Bitcoin protocol as an MDP and proposed a method using binary search to solve it approximately in order to find an optimal selfish mining strategy.
~\cite{ZurET20} improved the search method in~\cite{SapirshteinSZ16}, resulting in a method that finds an optimal strategy but with an order of magnitude less computation.
~\cite{squirl} and~\cite{RoozbehSB24} use deep reinforcement learning to automate discovery of attack strategies in Bitcoin and Nakamoto-PoStake, respectively, but without guarantees on optimality. 
Finally,~\cite{FerreiraHWY22} suggested modelling mining strategies in PoStake as an MDP and using an MDP solver to find optimal selfish mining strategies. However, due to the infinite size of their MDP, finding even an approximately optimal solution is undecidable. 


\section{Preliminaries}\label{sec:prelims}

\subsection{System Model}\label{sec:prelims_system}

We define a formal model for blockchain systems which we will consider throughout this work. 
The model describes how block mining proceeds and is parametrized by several parameters. 
Different parameter value instantiations then give rise to formal models for blockchains based on PoW, PoStake, or PoST protocols. We assume the miners in the blockchain protocol are either adversarial or honest. Honest miners all follow a prescribed protocol, whereas the adversarial miners work (i.e.~pool resources) together in a coalition.

\smallskip{\em System model.}
We assume that block mining proceeds in discrete time steps, as in many discrete-time models of PoW~\cite{PassSS17,GaziKR20} and PoStake~\cite{ouroboros} protocols. 
For a miner that owns a fraction $p \in [0,1]$ of the total resources in the blockchain and that can mine up to $k>0$ blocks at any given time step, we define $(p,k)$-mining as follows: the probability of mining a block at the given time step is proportional to $p \cdot k$, and the maximum number of blocks which are available for mining is $k$. One can think of $k$ as some further constraints on the number of blocks a miner can mine on at any given point in time, e.g.~due to the number of VDFs they own in PoST.
Hence, $(p,1)$-mining corresponds to mining in PoW based blockchains, $(p,k)$-mining for $k<\infty$ to mining in PoST based blockchains with $k$ VDFs, and $(p,\infty)$-mining to mining in PoStake based blockchains. 

\smallskip{\em Adversarial and broadcast model.} 
Let $p \in [0,1]$ denote the fraction of resources in the chain owned by the adversarial coalition.
We assume the adversarial coalition participates in $(p,k)$-mining, and the honest miners participate in $(1-p,1)$-mining, where the only block the honest miners mine on at any time step is the most recent block on the longest public chain.
When there is a tie, i.e., two longest chains are gossiped through the network, the honest miners will mine on the chain which is gossiped to them first. 
In such situations, $\gamma \in [0,1]$ denotes the {\em switching probability}, i.e.,~the probability of honest miners switching to the adversary's chain.  

\subsection{Selfish Mining Objective}\label{sec:prelim_sm}
\paragraph{Chain quality.}
Chain quality is a measure of the number of honest blocks in any consecutive segment of the blockchain.
A blockchain protocol is said to satisfy $(\mu,\ell)$-chain quality if for any segment of the chain of length $\ell$, the fraction of honest blocks in that segment is at least $\mu$~\cite{GarayKL15}.


\smallskip{\em Selfish mining objective.}
Let $\sigma$ denote an adversarial mining strategy. 
The selfish mining objective we analyse in our work is the \emph{expected relative revenue} of the adversary. 
Formally, let $\textsc{revenue}_{\mathcal{A}}$ and $\textsc{revenue}_{\mathcal{H}}$ denote the number of adversarial and honest blocks in the main chain respectively.
The expected relative revenue ($\RER$) of the adversary under strategy $\sigma$ is defined as
\[ \RER(\sigma) = \mathbb{E}^{\sigma}\Big[ \frac{\textsc{revenue}_{\mathcal{A}}}{\textsc{revenue}_{\mathcal{A}} + \textsc{revenue}_{\mathcal{H}}} \Big] \]
Note that the chain quality of the blockchain under an adversarial mining strategy $\sigma$ is simply $1-\RER(\sigma)$, hence our selfish mining objective captures the expected change in the chain quality. 

\subsection{Markov Decision Processes}\label{sec:prelim_mdp}

As mentioned in Section~\ref{sec:intro}, we will reduce the problem of computing optimal selfish mining strategies for the expected relative revenue objective to solving MDPs with mean-payoff objectives. In what follows, we recall the necessary notions on MDPs and formally define the mean-payoff objectives. For a finite set $X$, we use $\mathcal{D}(X)$ to denote the set of all probability distributions over $X$.

\smallskip{\em Markov decision process.}
A {\em Markov Decision Process (MDP)}~\cite{Puterman94} is a tuple $\MDP = (\States, \Actions, \transition, s_0)$ where
\begin{itemize}[leftmargin=*]
    \item $\States$ is a finite set of {\em states}, with $s_0 \in \States$ being the {\em initial state},
    \item $\Actions$ is a finite set of {\em actions}, overloaded to specify for each state $s \in \States$ the set of {\em available actions} $\Actions(s) \subseteq \Actions$, and
    \item $\transition: \States \times \Actions \rightarrow \mathcal{D}(\States)$ is a {\em transition function}, prescribing to each $(s,a) \in \States\times\Actions$ a probability distribution over successor states.
\end{itemize}
A {\em strategy} in $\MDP$ is a recipe to choose an action given a history of states and actions, i.e.~it is a function $\sigma: (\States \times \Actions)^{\ast} \times \States \rightarrow \mathcal{D}(A)$. In general, strategy can use randomization and memory. A {\em positional strategy} uses neither randomization nor memory, i.e.~it is a function $\sigma: \States \rightarrow \Actions$. We denote by $\Sigma$ and $\Sigma^p$ the set of all strategies and all positional strategies in $\MDP$. Given a strategy $\sigma$, it induces a probability measure $\mathbb{P}^{\sigma}[\cdot]$ in $\MDP$ with an associated expectation operator $\mathbb{E}^{\sigma}[\cdot]$.

\smallskip{\em Mean-payoff objective.} A {\em reward function} in $\MDP$ is a function $r: \States\times\Actions\times\States \rightarrow \mathbb{R}$ which to each state-action-state triple prescribes a real-valued reward. Given an MDP $\MDP$, a reward function $r$ and a strategy $\sigma$, we define the {\em mean-payoff} of $\sigma$ with respect to $r$ via
\[ \MP(\sigma) = \mathbb{E}^{\sigma}\Big[ \liminf_{N \rightarrow \infty}\frac{\sum_{n=1}^N r_n}{N} \Big], \]
where $r_n = r_n(s_n,a_n,s_{n+1})$ is the reward incurred at step $n$. 

The mean-payoff MDP problem is to compute the maximal mean-payoff that a strategy in the MDP can achieve. A classic result in MDP analysis is that there always exists a positional strategy $\sigma^\ast$ that achieves maximal mean-payoff in the MDP, i.e.~there is $\sigma^\ast \in \Sigma^p$ such that $\MP(\sigma^\ast) = \max_{\sigma \in \Sigma^p} \MP(\sigma) = \sup_{\sigma \in \Sigma} \MP(\sigma)$~\cite{Puterman94}. Furthermore, the optimal positional strategy and the mean-payoff that it achieves can be efficiently computed in polynomial time~\cite{Puterman94,filar2012competitive}.

\section{Selfish Mining Attack}\label{sec:methodology}


\subsection{Overview}\label{sec:methoverview}

\paragraph{Motivation.} In order to motivate our selfish mining attack, we first recall the classic selfish mining attack strategy in Bitcoin~\cite{EyalS18}. Recall, the goal of the selfish mining strategy is to mine a private chain that overtakes the public chain (see Figure~\ref{fig:basic_sm}). Selfish miners secretly fork the main chain and mine privately, adding blocks to a private, unannounced chain. 
These blocks are only revealed when the length of the private chain catches up with that of the main chain, forcing the honest miners to switch from the main to the private chain and waste computational resources. 
While in PoW blockchains each party mines on one block, in blockchains based on efficient proof systems parties can mine on multiple blocks due to ease of generating proofs. 
Our selfish mining attack exploits this observation by creating multiple private forks concurrently.

\begin{figure*}[t]
\centering
\begin{subfigure}{.45\textwidth}
  \centering
  \includegraphics[width=0.8\linewidth]{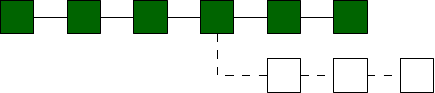}
  \caption{Classic selfish mining attack in Bitcoin.} 
  \label{fig:basic_sm}
\end{subfigure}%
\begin{subfigure}{.45\textwidth}
  \centering
  \centering
  \includegraphics[width=0.8\linewidth]{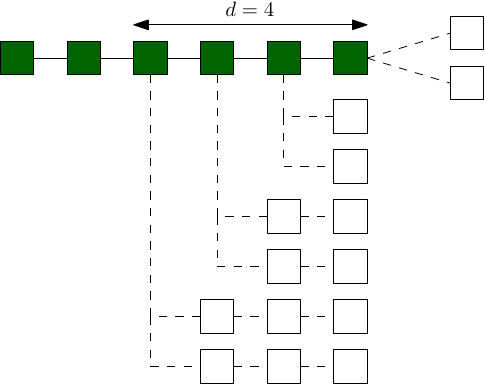}
  \caption{An illustration of our selfish mining attack with depth $d=4$ and $f=2$. Green boxes denote the main chain and white boxes denote potential blocks that can be mined under attack strategy.}
  \label{fig:attack}
\end{subfigure}
\label{fig:test}
\end{figure*}


\smallskip{\em Outline of the attack.} Our selfish mining attack proceeds with the adversary creating several private forks at different blocks in the main chain, see Figure~\ref{fig:attack} for an illustration. Rather than forking on the most recent block alone, the adversary creates up to $f$ forks on each of the last $d$ blocks on the main chain. Here, $f$ and $d$ are parameters of the attack where $f$ is the number of forks created at each public block and $d$ 
represents the depth of the adversary's attack on the chain.
One can view $d$ as the persistence parameter of the blockchain, which represents the depth at which earlier blocks are practically guaranteed to remain in the main chain~\cite{GarayKL15}. 
At each time step, the adversary can either perform a
\begin{enumerate}[leftmargin=*]
    \item {\em mining action}, i.e.~attempt to mine a new block, or
    \item {\em fork reveal action}, i.e.~publicly announce one of the private forks whose length is greater than or equal to that of the main chain. 
\end{enumerate}
Deciding on the optimal order of mining and fork reveal actions that the adversary should perform at each time step towards maximizing its expected relative revenue is a highly challenging problem. This is because, at each time step, the party to mine the next block is chosen probabilistically (see our System Model in Section~\ref{sec:prelim_sm}), and the process results in a system with an extremely large number of states.
For any given precision parameter $\epsilon > 0$, our analysis will provide both an \emph{$\epsilon$-optimal strategy} among selfish mining strategies that the adversary can follow, together with the \emph{exact value} of the expected relative revenue guaranteed by this strategy. 



\smallskip{\em Formal model of the attack.} The goal of our analysis is to find an optimal selfish mining strategy which maximizes the expected relative revenue of the adversary. Note that this is a sequential decision making problem, since the optimal strategy under the above selfish mining setting must at every decision step optimally choose whether to perform a mining or a fork release action. 
Hence, in order to analyze this problem, 
in Section~\ref{sec:methmodel} we formally model our problem as an MDP. The state space of the MDP consists of all possible configurations of the main chain and private forks with the initial state corresponding to the time step at which the selfish mining attack is initiated. The action space of the MDP consists of the mining action as well as one fork release action for each private fork whose length is greater than or equal to that of the main chain. Finally, the probabilistic transition function of the MDP captures the probabilistic process that generates and determines ownership (honest or adversarial) of new blocks to be added to the blockchain, as well as the process of determining the new main chain whenever the adversary publishes one or more private forks of equal length.

\smallskip{\em Formal analysis of the attack.} Recall, the objective of our selfish mining attack is to maximize the expected relative revenue of the adversary. 
To do this, in Section~\ref{sec:methanalysis} we define a class of reward functions in the MDP constructed in Section~\ref{sec:methmodel}. We show that, for any $\epsilon>0$, we can compute an $\epsilon$-optimal selfish mining strategy in the MDP and the exact value of the expected relative revenue that it guarantees by solving the mean-payoff MDP with respect to a reward function belonging to the class of constructed reward functions. We solve mean-payoff MDPs by using off-the-shelf tools as mentioned in Section~\ref{sec:prelim_mdp}. Our analysis yields a {\em fully automated} method for computing $\epsilon$-optimal strategies and values of the expected relative revenue that provides {\em formal guarantees} on the correctness of its results.

\subsection{Formal Model}\label{sec:methmodel}

We now formally model our selfish mining attack as an MDP. Recall, an MDP $\MDP = (\States, \Actions, \transition, s_0)$ is an ordered tuple where $\States$ is a finite set of states, $\Actions$ is a finite set of actions overloaded to specify for each state $s \in \States$ the set of available actions $\Actions(s) \subseteq \Actions$, $\transition: \States \times \Actions \rightarrow \mathcal{D}(\States)$ is a probabilistic transition function, and $s_0 \in \States$ is the initial state. Hence, in order to formally model our attack as an MDP, we need to define each of these four objects.

\smallskip{\em Model parameters.} Our MDP model uses as a basis the System Model that we defined in Section~\ref{sec:prelim_sm}. Thus, it is parameterized by the parameters $p$ and $\gamma$, but also by three additional parameters specific the selfish mining attack itself. We use $\mathbb{N}$ to denote the set of all positive integers:
\begin{itemize}[leftmargin=*]
\item {\em Relative resource of the adversary.} $p \in [0,1]$ denotes the fraction of resources in the blockchain owned by the adversary.
\item {\em Switching probability.} $\gamma \in [0,1]$ denotes the probability of an honest miner switching to a newly revealed adversarial chain that is of the same length as the main chain. 
\item {\em Attack depth.} $d \in \mathbb{N}$ denotes the depth of the adversary's attack on the chain, i.e.~the number of last blocks on the main chain on which the adversary mines private forks.
\item {\em Forking number.} $f \in \mathbb{N}$ denotes the number of private forks that the adversary creates at each of the last $d$ public blocks.
\item {\em Maximal fork length.} $ l \in \mathbb{N}$ denotes the maximal private fork length. Introducing this parameter ensures that our MDP model consists of finitely many states, which is necessary since existing mean-payoff MDP solvers are applicable to finite state MDPs~\cite{storm,prism}. We discuss the implications of this in Section~\ref{sec:methdiscussion}.
\end{itemize}

{\em MDP definition.} We define the MDP $\MDP = (\States, \Actions, \transition, s_0)$ as follows:
\begin{itemize}[leftmargin=*]
    \item {\em States.} The state space is defined via
    \begin{equation*}
    \begin{split}
        S = \Big\{ (\Config, \Own, \type)\, \Big|\, &\Config \in \{0,\dots,l\}^{d \times f}, \Own \in \{
    \textsc{honest},\textsc{adversary}\}^{d-1}, \\
    &\type \in \{\textsc{mining},\textsc{honest},\textsc{adversary}\}\Big\},
    \end{split}
    \end{equation*}
    i.e.~each state is defined as a triple $(\Config, \Own, \type)$ where $\Config$ defines the current blockchain configuration (i.e.~topology) up to depth $d$, $\Own$ specifies who owns each block in the main chain up to depth $d$ (honest miners or the adversary), and $\type$ specifies whether the parties are still mining or some party has mined a new block and is supposed to add it to the blockchain. In particular:
    \begin{itemize}
        \item For each $1 \leq i \leq d$ and $1 \leq j \leq f$, we use $\Config[i,j]$ to denote the length of the $j$-th private fork mined by the adversary at the depth $i$ block on the main chain. Each private fork has length at most $l$, so we impose that each $C[i,j] \in \{0,\dots,l\}$.
        \item For each $1 \leq i < d$, we use $\Own[i]$ to denote who owns the block at depth $i$ in the main chain. In particular, $\Own[i] = \textsc{honest}$ if the block is owned by honest miners and $\Own[i] = \textsc{adversary}$ if the block is owned by the adversary.
        \item Finally, $\type$ specifies whether a new block to be added to the blockchain is still being mined in which case we set $\type = \textsc{mining}$, or if some party has generated the proof and gets to add a new block in which case we set $\type = \textsc{honest}$ and $\type = \textsc{adversary}$, respectively.
    \end{itemize}

    \item {\em Initial State.} Initial state $s_0 = (\Config_0, \Own_0, \type_0)$ corresponds to the time step at which the selfish mining attack is initiated. Hence, we set $\Config_0 = [0]^{d \times f}$ since the length of each private fork is initially $0$, $\Own_0 = [\textsc{honest}]^d$ and $\type_0 = \textsc{mining}$.

    \item {\em Actions.} The action space is defined via
    \[ A = \{\textsc{mine}\} \cup \Big\{ \textsc{release}_{i,j,k} \,\Big|\, 1 \leq i\leq d, 1 \leq j \leq f, 1 \leq k \leq l \Big\}. \]
    Intuitively, $\textsc{mine}$ prescribes that the adversary should not release any private forks and should simply continue mining new blocks. On the other hand, $\textsc{release}_{i,j,k}$ prescribes that the adversary should publish the first $k$ blocks of the $j$-th private fork mined on the block at depth $i$ in the main chain.
    The set of available actions at each MDP state $s = (\Config,\Own,\type)$ is defined as follows:
    \begin{itemize}
        \item If $\type = \textsc{mining}$, then $\Actions(s) = \{\textsc{mine}\}$.
        \item If $\type \in \{\textsc{honest},\textsc{adversary}\}$, then \\ $\Actions(s) = \{\textsc{mine}\} \cup \{ \textsc{release}_{i,j,k} \,|\, k \leq \Config[i,j] \}$,
        where the condition $k \leq \Config[i,j]$ simply ensures that the length of the published part of the private fork cannot exceed the total length of the private fork.
    \end{itemize}
    
    \item {\em Transition Function.} Finally, we define the transition function $P: \States \times \Actions \rightarrow \mathcal{D}(S)$. Let $s = (\Config,\Own,\type)$ and $a \in \Actions(s)$ be an action available in $s$:
    \begin{itemize}
        \item If $\type = \textsc{mining}$, then we must have $a = \textsc{mine}$ as the only available action. The next state in the MDP is chosen probabilistically, based on who mines the next block and where. Recall, the honest miners own $1-p$ fraction of resources in the blockchain and are mining on the main chain only, whereas the adversary owns $p$ fraction of resources but mines on at most $d \cdot f$ private forks. Note that the adversary will concurrently mine on top of each private fork starting at public blocks up to depth $d$ in the main chain, as well as on top of each public block up to depth $d$ at which not all $f$ private works were initiated. 
        Hence:
        \begin{itemize}
            \item Denote by $\sigma$ the total number of blocks that the adversary is mining on. For each $1 \leq i \leq d$ and $1 \leq j \leq f$, if the $j$-th private fork on the public block at depth $i$ in the main chain is not empty, adversary mines a new block on it and the MDP moves to $s_{\textsc{adv}}^{i,j} = (\Config', \Own,\textsc{adversary})$ with probability
            \[ P(s,a)(s_{\textsc{adv}}^{i,j}) = \frac{p}{1 - p + p \cdot \sigma}. \]
            Here, $\Config'$ coincides with $\Config$ on all entries except 
            $\Config'[i, j] = \min\{\Config[i, j] + 1,l\}$, where 
            the new block is found on top of the $j$-th private 
            fork on the block at depth $i$. 
            The minimum ensures that new block is not added if the length of the fork would exceed $l$.
            
            Moreover, if at least one private fork on the public block at depth $i$ in the main chain is empty, adversary mines a block on it to start a new private fork and the MDP moves to $s_{\textsc{adv}}^{i,j^{\ast}} = (\Config'', \Own,\textsc{adversary})$ with probability $\frac{p}{1 - p + p \cdot \sigma}$, where $j^{\ast}$ is the smallest index of a private fork that is currently empty at the depth $i$ public block and $\Config''$ coincides with $\Config$ on all entries except $\Config''[i, j^{\ast}] = 1$. 
            
            \item Honest miners add a new block to the main chain and the MDP moves to state $s_{\textsc{honest}} = ([0]^{f}\cdot\Config[1:d-1],[\textsc{honest}]\cdot\Own[1:d-2],\textsc{honest})$ with probability
            \[ P(s,a)(s_{\textsc{honest}}) = \frac{1-p}{1 - p + p \cdot \sigma}. \] 
            where $\sigma$ is as defined above. We use $[0]^{f}\cdot\Config[1:d-1]$ to denote the $d \times f$ matrix whose first column consists of zeros (representing empty private forks on the newly added block) and the remaining columns are the first $d-1$ columns of $\Config$. Similarly, $[\textsc{honest}]\cdot\Own[1:d-2]$ denotes the vector whose first component is $\textsc{honest}$ corresponding to the new block, followed by the first $d-2$ components of $\Own$.
        \end{itemize}
        
        \item If $\type \neq \textsc{mining}$ but $a = \textsc{mine}$, then the adversary continues mining new blocks and the MDP moves to $s_{\textsc{mine}} = (\Config,\Own,\textsc{mining})$ with probability
        \[ P(s,a)(s_{\textsc{mine}}) = 1. \]

        \item If $\type = \textsc{honest}$ and $a = \textsc{release}_{i,j,k}$, then the adversary publishes the first $k$ blocks of the $j$-th private fork mined on the block at depth $i$ in the main chain. Hence, the next MDP state is determined by the chain which gets accepted as the main chain by honest miners. If the newly published fork is strictly longer than the main chain, the honest miners accept it as the new main chain with probability $1$. Otherwise, if the main chain and the published fork have the same length, then a ``race'' between the chains happens and the published fork becomes the main chain with switching probability $\gamma$. Hence:
        \begin{itemize}
            \item If the published fork is longer than the main chain, the honest miners switch to the fork and the MDP moves to
            $s_{\textsc{accept}} = (\Config'\cdot[0]^{(min(d, k)-1)\times f}\cdot\Config'',[\textsc{adversary}]^{min(d-1, k)}\cdot\Own[k+1:d-1],\textsc{mining})$ with probability
            \[ P(s,a)(s_{\textsc{accept}}) = 1. \]
            Here, $\Config'$ is a $(1\times f)$-dimensional vector with $0$ entries except that $\Config'[1, 1] = \Config[i, j] - k$. So, the adversary keeps the part of the published private fork consisting of the last $\Config[i, j] - k$ blocks that have not been revealed.
            On the other hand, $\Config''$ coincides with $\Config[k+1:d]$ on all entries except that $\Config''[i, j] = 0$. Hence, the adversary keeps and continues mining on all usable private forks, with one new private fork of initial length $0$ being created to replace the published fork. Note, $\Own[k+1:d-1]$ and $\Config[k+1:d]$ are empty if $k + 1 > d-1$ and $k + 1 > d$, respectively.
            
            \item Otherwise, if the published fork and the main chain are of the same length, a race happens and honest miners switch to the published fork with probability~$\gamma$. and the new block will be mined by adversary with probability $p$.
            So, the MDP moves to $s_{\textsc{accept}} = (\Config'\cdot[0]^{(min(d, k)-1)\times f}\cdot\Config'',[\textsc{adversary}]^{min(d-1, k)}\cdot\Own[k+1:d-1],\textsc{mining})$ with probability
            \[ P(s,a)(s_{\textsc{accept}}) = \gamma, \]
            where $\Config'$ and $\Config''$ are as above, and moves to $s_{\textsc{reject}} = (\Config, \Own,\textsc{mining})$ with probability
            \[ P(s,a)(s_{\textsc{reject}}) = 1 - \gamma. \]
        \end{itemize}
        \item Finally, if $\type = \textsc{adversary}$, $a = \textsc{release}_{i,j,k}$ and the published private fork is longer than the main chain, the MDP moves to \\
        $s_{\textsc{accept}} = (\Config'\cdot[0]^{(min(d, k)-1)\times f}\cdot\Config'',[\textsc{adversary}]^{min(d-1, k)}\cdot\Own[k+1:d-1],\textsc{mining})$ 
        with probability
            \[ P(s,a)(s_{\textsc{accept}}) = 1. \]
        where $\Config'$ and $\Config''$ are defined as the previous case.
        Note that if the published fork and the main chain are of the same length, the race cannot happen as the last block was mined by the adversary and enough time has passed for all the miners to receive the public chain.
    \end{itemize}
\end{itemize}

\subsection{Formal Analysis}\label{sec:methanalysis}

\paragraph{Goal of the analysis.} We now show how to compute the optimal expected relative revenue together with an adversarial strategy achieving it in the MDP $\MDP = (\States, \Actions, \transition, s_0)$, up to an arbitrary precision parameter $\epsilon>0$. Formally, for each strategy $\sigma$ in $\MDP$, let
\[ \RER(\sigma) = \mathbb{E}^{\sigma}\Big[ \frac{\textsc{revenue}_{\mathcal{A}}}{\textsc{revenue}_{\mathcal{A}} + \textsc{revenue}_{\mathcal{H}}} \Big] \]
be the expected relative revenue under adversarial strategy $\sigma$, i.e.~the relative ratio of the number of blocks accepted on the main chain belonging to the adversary and to honest miners. Moreover, let
\[ \RER^\ast = \sup_{\text{strategy } \sigma \text{ in } \MDP} \RER(\sigma) \]
be the optimal expected relative revenue that an adversarial strategy can attain. Given a precision parameter $\epsilon>0$, our goal is to compute
\begin{enumerate}[leftmargin=*]
    \item a lower bound $\overline{\RER} \in [\RER^\ast-\epsilon,\RER^\ast]$, and
    \item a strategy $\overline{\sigma}$ in $\MDP$ such that $\RER(\overline{\sigma}) \in [\RER^\ast-\epsilon,\RER^\ast]$.
\end{enumerate}
We do this by defining a class of reward functions in the MDP  $\MDP$ and showing that, for any value of the precision parameter $\epsilon>0$, we can compute the above by solving the mean-payoff MDP with respect to a reward function belonging to this class. Our analysis draws insight from that of~\cite{SapirshteinSZ16}, which considered selfish mining in PoW blockchains and also reduced reasoning about expected relative revenue to solving mean-payoff MDPs with respect to suitably defined reward functions. However, in contrast to~\cite{SapirshteinSZ16}, we consider selfish mining in efficient proof systems in which the adversary can mine on multiple blocks, meaning that our design of reward functions and the analysis require additional care.

\smallskip{\em Reward function definition.} The key challenge in designing the reward function is that the main chain and the blocks on it may change whenever the adversary publishes a private fork. Hence, we design the reward function to incur positive (resp.~negative) reward whenever a block owned by the adversary (resp.~honest miners) is accepted at the depth strictly greater than $d$ in the main chain. Since the adversary only mines and publishes private forks mined on blocks up to depth $d$ in the main chain, this means that blocks beyond depth $d$ are {\em guaranteed} to remain on the main chain.

Formally, for each $\beta \in [0,1]$, we define $r_\beta: \States \times \Actions \times \States \rightarrow \mathbb{R}$ to be a reward function in $\MDP$ which to each state-action-state triple $(s,a,s')$ assigns the reward:
\begin{itemize}[leftmargin=*]
    \item $1-\beta$, for each block belonging to {\em the adversary} accepted at depth greater than $d$ as a result of performing the action;
    \item $-\beta$, for each block belonging to {\em honest miners} accepted at depth greater than $d$ as a result of performing the action.
\end{itemize}
This definition can be formalized by following the same case by case analysis as in the definition of the transition function in Section~\ref{sec:methmodel}. For the interest of space, we omit the formal definition. For each $\beta$ and strategy $\sigma$ in the MDP, let $\MP_{\beta}(\sigma)$ be the mean-payoff under $\sigma$ with respect to the reward function $r_\beta$, and let $\MP_{\beta}^\ast = \sup_{\sigma}\MP_{\beta}(\sigma)$. Note that, since for each state-action-state triple the value of the reward $r_\beta(s,a,s')$ is monotonically decreasing in $\beta \in [0,1]$, we have that $\MP_\beta^\ast$ is also monotonically decreasing in $\beta \in [0,1]$.

\smallskip{\em Formal analysis.} Our formal analysis is based on the following theorem. For clarity of exposition, we defer the proof of the theorem to Appendix~\ref{app:proof_main} . For every $\epsilon > 0$, the theorem shows how to relate the optimal expected relative revenue in the MDP and $\epsilon$-optimal strategies to the optimal mean-payoff and $\epsilon$-optimal strategies under the reward function $r_\beta$ for a suitably chosen value of $\beta$.

\begin{restatable}[]{theorem}{main}
\label{thm:sound}
We have $\MP^\ast_{\beta^\ast} = 0$  if and only if $\beta^\ast = \RER^\ast$. Moreover, if $\epsilon > 0$ and $\beta \in [\RER^\ast-\epsilon,\RER^\ast]$, then for any strategy $\sigma$ such that $\MP_{\beta}(\sigma) = \MP^\ast_{\beta}$ we have $\RER(\sigma) \in [\RER^\ast-\epsilon,\RER^\ast]$.
\end{restatable}


\begin{algorithm}[t]
\caption{Formal analysis procedure}
\label{alg:algorithm}
\begin{algorithmic}
\STATE \textbf{Input} Precision parameter $\epsilon > 0$, MDP parameters $p,d,f,l,\gamma$
\STATE $\MDP \leftarrow$ MDP constructed as in Section~\ref{sec:methmodel}
\STATE $\beta_{\text{low}}, \beta_{\text{up}} \leftarrow 0, 1$
\WHILE{$\beta_{\text{up}} - \beta_{\text{low}} \geq \epsilon$}
\STATE $\beta \leftarrow (\beta_{\text{low}} + \beta_{\text{up}})/2$
\STATE $\MP_\beta^\ast, \sigma_\beta \leftarrow$ solve mean-payoff MDP $\MDP$ with reward $r_\beta$
\IF{$\MP_\beta^\ast < 0$}
\STATE $\beta_{\text{up}} \leftarrow \beta$
\ELSE
\STATE $\beta_{\text{low}} \leftarrow \beta$
\ENDIF
\ENDWHILE
\STATE $\overline{\RER} \leftarrow \beta_{\text{low}}$
\STATE $\overline{\sigma} \leftarrow$ solve mean-payoff MDP $\MDP$ with reward $r_{\beta_{\text{low}}}$
\STATE \textbf{Return} Expected relative revenue $\overline{\RER}$, strategy $\overline{\sigma}$
\end{algorithmic}
\end{algorithm}

Following Theorem~\ref{thm:sound}, we compute $\overline{\RER}$ and $\overline{\sigma}$ for a given precision $\epsilon > 0$ as follows. Algorithm~\ref{alg:algorithm} shows the pseudocode of our formal analysis. The algorithm performs binary search in $\beta \in [0,1]$ in order to find a value of $\beta$ for which $\beta \in [\RER^\ast-\epsilon,\RER^\ast]$. In each iteration of the binary search, the algorithm uses an off-the-shelf mean-payoff MDP solver to compute the optimal mean-payoff $\MP^\ast_{\beta}$ and a strategy $\sigma_\beta$ attaining it. Upon binary search termination, we have $\RER^\ast \in [\beta_{\text{low}}, \beta_{\text{up}}]$, since $\MP^\ast_{\beta}$ is a monotonically decreasing function in $\beta$ and since $\MP^\ast_{\beta^\ast} = 0$ if and only if $\beta^\ast = \RER^\ast$ by the first part of Theorem~\ref{thm:sound}. Hence, as the binary search terminates when $\beta_{\text{up}} - \beta_{\text{low}} < \epsilon$, we conclude that $\overline{\RER} = \beta_{\text{low}} \in  [\RER^\ast-\epsilon,\RER^\ast]$. Moreover, since $\overline{\sigma}$ is optimal for the mean-payoff objective with reward $r_{\beta_{\text{low}}}$ and since $\beta_{\text{low}} \in  [\RER^\ast-\epsilon,\RER^\ast]$, by the second part of Theorem~\ref{thm:sound} we have $\RER(\overline{\sigma}) \in [\RER^\ast-\epsilon,\RER^\ast]$.

\begin{corollary}[Correctness of the analysis]
    Let $\epsilon>0$. Suppose that Algorithm~\ref{alg:algorithm} returns a value $\overline{\RER}$ and a strategy $\overline{\sigma}$ in $\MDP$. Then, we have $\overline{\RER}, \RER(\overline{\sigma}) \in [\RER^\ast-\epsilon,\RER^\ast]$.
\end{corollary}

\begin{corollary}[Formal lower bound]
    Suppose that Algorithm~\ref{alg:algorithm} returns a value $\overline{\RER}$ and a strategy $\overline{\sigma}$. Then, there exists a selfish mining attack in the blockchain that achieves the expected relative revenue of at least $\overline{\RER}$.
\end{corollary}

\subsection{Key Features and Limitations}\label{sec:methdiscussion}


\paragraph{Key features.} The key features of our selfish mining attack and formal analysis are as follows:
\begin{enumerate}[leftmargin=*]
    \item {\em Fully automated analysis.} Manual (i.e.~non-automated) analysis of optimal selfish mining attacks is already challenging and technically involved for PoW blockchains, where the adversary only grows a single private fork~\cite{EyalS18}. Hence, it would be even more difficult and potentially intractable in blockchains based on efficient proof systems. By modelling our selfish mining attack as an MDP and reducing the analysis to solving mean-payoff MDPs, we leverage existing methods for formal analysis of MDPs to obtain a {\em fully automated analysis} procedure, thus avoiding the necessity for tedious manual analyses.
    \item {\em Formal guarantees on correctness.} Our analysis provides {\em formal guarantees} on the correctness of its output. Again, this is achieved by formally reducing our problem to solving mean-payoff MDPs for which exact algorithms with formal correctness guarantees are available~\cite{storm,prism}.
    \item {\em Flexibility of the analysis.} Our analysis is agnostic to the values of system model and attack parameters and it is {\em flexible} to their changes. Hence, it allows us to tweak the parameter values and study their impact on the optimal expected relative revenue, while preserving formal guarantees on the correctness. To illustrate the flexibility, observe that:
    \begin{itemize}
        \item If the attack depth $d$, forking number $f$ or maximal fork length $l$ of the attack change, then both the state space and the action space of the MDP change.
        \item If the relative resource of the adversary $p$ or the switching probability $\gamma$ change, then the transition function of the MDP changes.
        \item As we show in our experiments in Section~\ref{sec:eval}, a change in any of these parameter values results in a change in the optimal expected relative revenue that the adversary can achieve.
    \end{itemize}
    The flexibility of our analysis is thus a significant feature, since it again avoids the need for tedious manual analyses for different parameter values that give rise to different MDPs. 
\end{enumerate}

\smallskip{\em Limitations.} While our formal analysis computes an optimal selfish mining strategy in the MDP up to a desired precision, note that there still exist selfish mining attacks that do not correspond to any strategy in our MDP model. Hence, the strategy computed by our method is optimal only with respect to the {\em subclass} of strategies captured by the MDP model. There are two key reasons behind the incompleteness of our MDP model:
\begin{enumerate}[leftmargin=*]
    \item {\em Bounded forks.} In order to ensure finiteness of our MDP model, we impose an upper bound $l$ on the maximal length of each private fork. This means that the adversary cannot grow arbitrarily long private forks. Since the probability of the adversary being able to grow extremely long private forks is low, we believe that this limitation does not significantly impact the expected relative revenue of selfish mining strategy under this restriction.
    \item {\em Disjoint forks vs fork trees.} Our attack grows private forks on different blocks in the main chain. However, rather than growing multiple disjoint private forks, a more general class of selfish mining attacks would be to allow growing {\em private trees}. We stick to disjoint private forks in order to preserve {\em computational efficiency} of our analysis, since allowing the adversary to grow private trees would result in our MDP states needing to store information about each private tree topology, which would lead to a huge blow-up in the size of the MDP. In contrast, storing disjoint private forks only requires storing fork lengths, resulting in smaller MDP models.
\end{enumerate}
We conclude by noting that, while our formal analysis is incomplete due to considering a subclass of selfish mining attacks, the formal guarantees provided by our analysis still ensure that we compute a {\em true lower bound} on the expected relative revenue that a selfish mining attack achieves.

\section{Experimental Evaluation}\label{sec:eval}

\newcommand{\RQ}[1]{\bf{RQ#1}}

We implement the MDP model and the formal analysis procedure presented in Section~\ref{sec:methodology} and perform an experimental evaluation towards answering the following research questions (RQs):
\begin{enumerate}[label=\RQ{\arabic*}]
    \item What is the expected relative revenue that our selfish mining strategy achieves? How does it compare to direct extensions of classic selfish mining attacks on PoW blockchains~\cite{EyalS18} or to mining honestly?
    \item How do different values of the System Model parameters impact the expected relative revenue that our selfish mining attack can achieve? The System Model parameters include the relative resource of the adversary $p \in [0,1]$ and the switching probability $\gamma \in [0,1]$.


\end{enumerate}

\begin{figure*}[t]
    \centering
    \begin{subfigure}[b]{0.31\linewidth}
        \includegraphics[width=\linewidth]{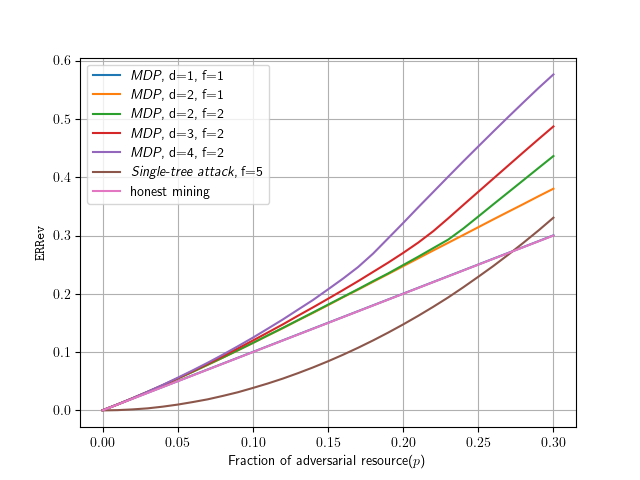}
        \caption{$\gamma=0$}
        \label{fig:exp_0}
    \end{subfigure}
    \begin{subfigure}[b]{0.31\linewidth}
        \includegraphics[width=\linewidth]{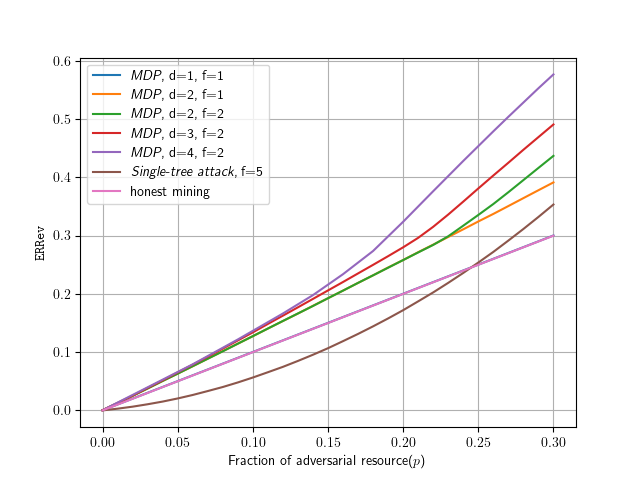}
        \caption{$\gamma=0.25$}
        \label{fig:exp_25}
    \end{subfigure}
    \begin{subfigure}[b]{0.31\linewidth}
        \includegraphics[width=\linewidth]{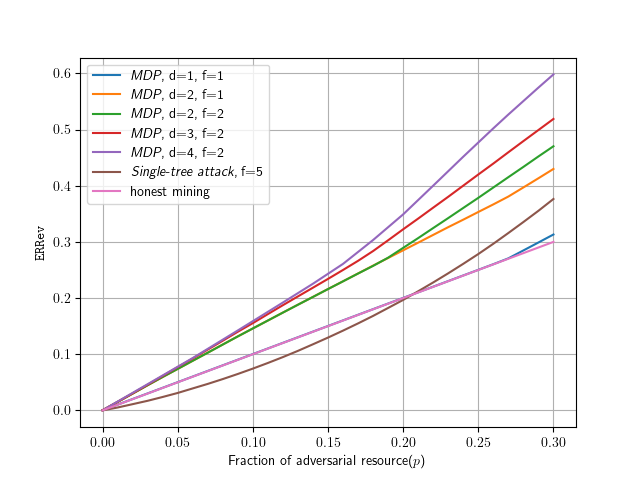}
        \caption{$\gamma=0.5$}
        \label{fig:exp_50}
    \end{subfigure}
    \begin{subfigure}[b]{0.31\linewidth}
        \includegraphics[width=\linewidth]{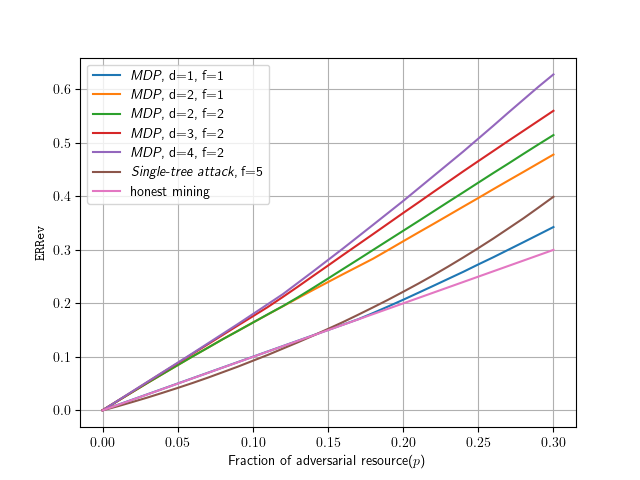}
        \caption{$\gamma=0.75$}
        \label{fig:exp_75}
    \end{subfigure}
    \begin{subfigure}[b]{0.31\linewidth}
        \includegraphics[width=\linewidth]{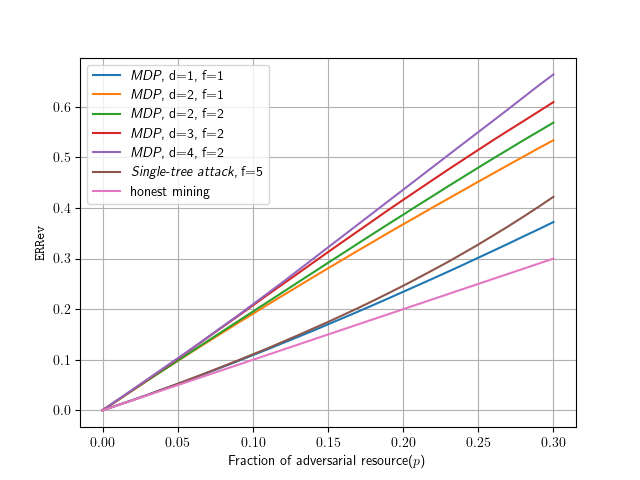}
        \caption{$\gamma=1$}
        \label{fig:exp_100}
    \end{subfigure}
    \caption{Comparison of expected relative revenue $\RER$  as a function of adversarial resource achieved by our attack and the baselines for different values of $\gamma$.}
    \label{fig:mdpplots}
\end{figure*}

{\em Baselines.} To answer these RQs, we compare our selfish mining attack against two baselines:
\begin{enumerate}[leftmargin=*]
    \item {\em Honest mining strategy.}
    This is the strategy extending only the leading block of the main chain. 
    \item {\em Single-tree selfish mining attack.} This is the attack strategy that exactly follows the classic selfish mining attack in Bitcoin proposed in~\cite{EyalS18}, however it grows a private tree fork rather than a private chain. Analogously as in~\cite{EyalS18}, the adversary publishes the private tree whenever the length of the main chain catches up with the depth of the private tree. We omit the formal model of this baseline due to space limitations. We also use this baseline to empirically evaluate how severe is the second limitation discussed in Section~\ref{sec:methdiscussion}.
\end{enumerate}

{\em Experimental setup.} We perform an experimental comparison of our attack and the two baselines for the values of the adversarial relative resource $p \in [0,0.3]$ (in increments of $0.01$) and the switching probability $\gamma \in \{0,0.25,0.5,0.75,1\}$. As for the parameters of each selfish mining attack:
\begin{itemize}[leftmargin=*]
    \item For our selfish mining attack, we set the maximal length of private forks $l=4$ and consider all combinations \\ $(d,f) \in \{(1,1), (2,1), (2,2), (3,2), (4,2)\}$ of the values of the attack depth $d$ and the forking number $f$.
    \item For the single-tree selfish mining attack baseline, we set the maximal depth of the private tree $l=4$ to match our maximal private fork length, and the maximal width of the private tree $f=5$. 
\end{itemize}
All experiments were run on Ubuntu 20, 2.80GHz 11th Gen Intel(R) Core(TM) i7-1165G7 CPU, 16 GB RAM, and 16 GB SWAP SSD Memory. For solving mean-payoff MDPs, we use the probabilistic model checker Storm~\cite{storm}, a popular MDP analysis tool within the formal methods community\footnote{ Refer to our github repository for our implementation details: \href{https://github.com/mehrdad76/Automated-Selfish-Mining-Analysis-in-EPS-Blockchains}{https://github.com/mehrdad76/Automated-Selfish-Mining-Analysis-in-EPS-Blockchains} }.

\begin{table}[htb!]
\centering
\begin{tabular}{ |p{4cm}||p{2cm}|p{1cm}|  }
 \hline
 Attack Type & Parameters &Time (s)\\
 \hline
 Our Attack   & $d=1,f=1$    &3.8\\
 Our Attack   & $d=2,f=1$    &5.4\\
 Our Attack   & $d=2,f=2$    &61.7\\
 Our Attack   & $d=3,f=2$    &2295.3\\
 Our Attack   & $d=4,f=2$    &77761.7\\
 Single-tree Selfish Mining   & $f=5$    &2406.8\\
 \hline
\end{tabular}
\caption{Runtimes for various attack types and parameter settings and $\gamma=0.5$.}
\label{tab:times}
\end{table}

\smallskip{\em Results.} 
Table~\ref{tab:times} shows the runtimes of both our selfish mining attack as well as the single-tree selfish mining attack given various parameter settings and for a fixed switching parameter of $\gamma=0.5$. We only show timings for $\gamma=0.5$ as we found the runtimes of our experiments to be very similar across all $\gamma$ parameter settings.
As can be seen from Table~\ref{tab:times}, increasing the depth of the attack increases the runtime of our evaluation by an order of magnitude due to the exponential increase in state space.

Experimental results are shown in Figure~\ref{fig:mdpplots}, showing plots for each $\gamma \in \{0,0.25,0.5,0.75,1\}$.
As we can see from the plots, our selfish mining attack consistently achieves {\em higher expected relative revenue} $\RER$ than both baselines for each value of $\gamma$, except when $d=1$ and $f=1$. Indeed, already for $d=2$ and $f=1$ when the adversary grows a single private fork on the first two blocks in the main chain, our attack achieves higher $\RER$ than both baselines. This shows that growing private forks at two different blocks already provides a more powerful attack than growing a much larger private tree at a single block. Hence, our results indicate that growing disjoint private forks rather than trees is not a significant limitation, justifying our choice to grow private forks towards making the analysis computationally tractable.

The attained $\RER$ grows significantly as we increase $d$ and $f$ and allow the adversary to grow more private forks. In particular, for $d=4$, $f=2$, and relative adversarial resource $p=0.3$, our attack achieves $\RER$ that is larger by at least $0.2$ than that of both baselines, for all values of the switching probability $\gamma$. This indicates a {\em significant advantage} of selfish mining attacks in efficient proof systems blockchains compared to PoW, as the ability to simultaneously grow multiple private forks on multiple blocks translates to a much larger $\RER$. Our results suggest that further study of techniques to reduce the advantage of the adversary when mining on several blocks is important in order to maintain reasonable chain quality for efficient proof systems blockchains.

Finally, we notice that larger $\gamma$ values correspond to larger $\RER$ in our strategies. 
This is expected, as larger $\gamma$ values introduce bias in the likelihood of the adversarial chain becoming the main chain. 
This is most pertinently observed in the case of  $d=f=1$: since $d=f=1$ corresponds to a strategy that only mines a private block on the leading block in the main chain, the only way to deviate from honest mining is to withhold a freshly mined block and reveal it together with the occurrence of a freshly mined honest block. As we can see in the plots, for $\gamma<0.5$ the achieved $\RER$ of the strategy with $d=f=1$ corresponds to that of honest mining and the two lines in plots mostly overlap, whereas this strategy only starts to pay off for $\gamma>0.5$ and for the proportion of resource $p>0.25$.
Altogether, this suggests that further and careful analysis of the control of the adversary over the broadcast network as well as the fork choice breaking rule is necessary. 

\smallskip{\em Key takeaways.} The key takeaways of our experimental evaluation are as follows:
\begin{itemize}[leftmargin=*]
    \item Our selfish mining attack achieves {\em significantly higher} $\RER$ than both baselines, reaching up to $0.2$ difference in $\RER$. Thus, our results strongly suggest that growing private forks at multiple blocks is much more advantageous than growing all forks on the first block in the main chain.
    \item Our results suggest that growing private trees rather than disjoint private forks would {\em not} lead to a significant improvement in the adversary's $\RER$. Hence, the second limitation of our attack discussed in Section~\ref{sec:methdiscussion} does not seem to be significant.
    \item Our results suggest that enhancing security against selfish mining attacks in efficient proof system blockchains requires further and careful analysis of the control that the adversary has over the broadcast system. In particular, for large values of the switching probability $\gamma$, even the simplest variant of our attack with $d=1$ and $f=1$ starts to pay off when $p>0.25$.
\end{itemize}

\section{Conclusion}\label{sec:conclusion}
We initiated the study of optimal selfish mining strategies for unpredictable blockchain protocols based on efficient proof systems. To this end, we considered a selfish mining objective corresponding to changes in chain quality and proposed a novel selfish mining attack that aims to maximize this objective. We formally modeled our attack as an MDP strategy and we presented a formal analysis procedure for computing an $\epsilon$-tight lower bound on the optimal expected relative revenue in the MDP and a strategy that achieves it for a specified precision $\epsilon>0$. The procedure is fully automated and provides formal guarantees on the correctness of the computed bound.

We believe that our work opens several exciting lines for future research. 
We highlight two particular directions.
First, our formal analysis only allows us to compute {\em lower bounds} on the expected relative revenue that an adversary can achieve. 
A natural direction of future research would be to consider computing {\em upper bounds} on the optimal expected relative revenue for fixed resource amounts.
Second, as discussed in Section~\ref{sec:methdiscussion}, our formal analysis only computes $\epsilon$-tight lower bounds on the expected relative revenue by following a strategy in our MDP model. However, our model in Section~\ref{sec:methmodel} introduces assumptions such as growing private forks instead of trees and bounding the maximal length of each fork for tractability purposes. It would be interesting to study whether these assumptions could be relaxed while still providing formal correctness guarantees. 

\begin{acks}
This work was supported in part by the ERC-2020-CoG 863818 (FoRM-SMArt) grant and the MOE-T2EP20122-0014 (Data-Driven Distributed Algorithms) grant.
\end{acks}

\bibliographystyle{ACM-Reference-Format}
\bibliography{sample-base.bib} 


\newpage
\appendix
\section{NaS Mining Objectives}\label{app:sm_objectives}
Adversarial mining strategies in blockchains based on efficient proof systems can be analysed with respect to several adversarial goals. 
Here, we outline three such goals: double spending, short and long term selfish-mining.

The first objective is double spending, where one 
considers the probability of an adversarial chain overtaking the public, honest chain~\cite{WangKBKOTV19,CohenP23} (see Figure~\ref{fig:overtake}). 
Here the goal of the adversary is to rapidly and secretly grow a sufficiently long private chain such that this private chain eventually overtakes the honest chain, this way removing a presumably confirmed transaction.  
What ``sufficiently'' long means depends on the confirmation time in the chain, e.g., in Bitcoin one generally assumes a transaction that is six blocks deep in the chain to be confirmed.

The second objective, ``short-term selfish mining'',  considered eschews the goal of overtaking the honest chain completely and focuses simply on finding an adversarial mining strategy that is more profitable for the adversary rather than following the stipulated mining protocol~\cite{FanZ17} (see Figure~\ref{fig:reward}). 
The profitability of an adversarial mining strategy under this objective is measured by the \emph{total} number of adversarial blocks on the main chain. Like in the analyses of selfish mining strategies under the first objective, analyses of strategies under this second objective also focus on finding the largest fraction of adversarial resources the blockchain can tolerate in order to be secure under such adversarial strategies.

The final objective, ``long-term selfish mining'',  considered is directly related to attacking the chain quality of the underlying blockchain. Under this objective, the adversary's reward is not measured by the total number of adversarial blocks on the main chain, but the \emph{relative} number of adversarial blocks on the main chain. 
Figure~\ref{fig:cq} illustrates the difference between simply maximising the total reward as per objective $2$ and the relative reward. In the top chain in Figure~\ref{fig:cq}, the adversary's total reward is $4$ which is larger than the total reward of $3$ in the bottom chain. However, the relative reward of the adversary in the top chain is $\frac{2}{3}$ which is smaller than the relative reward of $1$ in the bottom chain. 
Note that the chain quality and the relative reward sum up to $1$, thus maximising the relative reward minimises the chain quality.


\begin{figure}[t]
\minipage{0.45\textwidth}
  \includegraphics[width=0.5\linewidth]{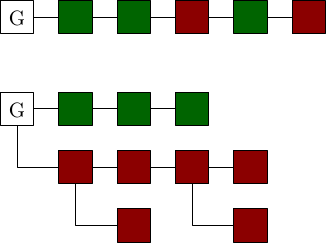}
  \caption{
  The top chain shows adversarial (red) and honest (green) miners extending the public chain. The bottom chain shows adversarial miners growing a private chain to overtake the public chain.
  }\label{fig:overtake}
\endminipage\hfill
\minipage{0.45\textwidth}
  \includegraphics[width=0.7\linewidth]{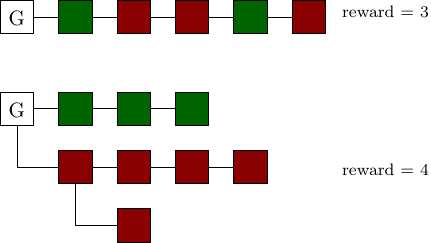}
  \caption{
  The chains show the reward of the adversary when following the honest strategy and when growing a private chain to overtake the public chain.}\label{fig:reward}
\endminipage\hfill
\minipage{0.45\textwidth}%
  \includegraphics[width=\linewidth]{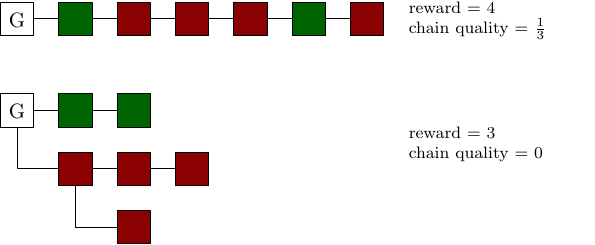}
  \caption{The chains show the reward of the adversary and chain quality of the chain when all users mine honestly and when the adversary grows a private chain.}\label{fig:cq}
\endminipage
\end{figure}

\section{Efficient Proof Systems}\label{sec:efficientproofsystems}

\paragraph{Proof of stake.} 
PoStake is a block leader election protocol where a leader is selected with probability proportionate to the amount of stake (i.e., coins) they hold in the ledger at the selection time.
Thus, a user with $p \in [0,1]$ fraction of stake is elected with probability proportionate to~$p$.
Examples of longest-chain blockchains based on PoStake are Ouroboros~\cite{ouroboros} and post-merge Ethereum~\cite{ethmerge}.

\smallskip{\em Proofs of space and time.}
Proof of Space (PoSpace) is a protocol between a \emph{prover} and a \emph{verifier} whereby the prover stores some data and, upon a challenge from the verifier, has to return a solution to the challenge that involves reading a small portion of the data. 
The consensus protocol of blockchains based on PoST~\cite{CohenP23} use both PoSpace challenges as well as verifiable delay functions~\cite{BonehBBF18,PietrzakVDF,Wesolowski18} (VDFs).
VDFs are functions that are inherently sequential to compute but the correctness of computation is efficiently verifiable. 
As such, the process of mining blocks in such blockchains depends not only on the amount of space allocated to compute PoSpace challenges, but also on the amount of VDFs to compute VDF challenges. 

\section{Proof of Theorem~\ref{thm:sound}}\label{app:proof_main}

\main*

\begin{proof}
    Before we prove the theorem, we start by making several observations and introducing additional notation. The proof assumes familiarity with basic notions of Markov chains and MDPs, for which we refer the reader to~\cite{norris,Puterman94}. First, we observe that every strategy in the MDP $\MDP = (\States,\Actions,\transition,s_0)$ gives rise to an {\em ergodic Markov chain}. To see this, observe that the initial state $s_0$ is reached with positive probability from any other state in the MDP, since honest miners with positive probability mine and add new blocks to the main chain for $d$ consecutive time steps upon which the MDP would return to the initial state $s_0$.
    
    Next, we define two auxiliary reward functions in the MDP $\MDP$:
    \begin{itemize}[leftmargin=*]
        \item $r^H:\States \times \Actions \times \States \rightarrow \mathbb{R}$, which incurs reward $1$ for each honest miner's block accepted at depth $> d$ upon the action.
        \item $r^A:\States \times \Actions \times \States \rightarrow \mathbb{R}$, incurs reward $1$ for each adversary's block accepted at depth $> d$ upon the action.
    \end{itemize}
    For each integer $n \geq 1$, denote by $r^A_n$ and $r^H_n$ the rewards incurred at the $n$-th step in the MDP. Since the expected relative revenue is defined as the expected ratio of the revenues of the adversary and the total revenue of all parties, for every strategy $\sigma$ we have
    \begin{equation*}
        \RER(\sigma) = \mathbb{E}^{\sigma}\Big[\liminf_{N\rightarrow\infty} \frac{\sum_{n=1}^N r^A_n}{\sum_{n=1}^N r^A_n + r^H_n} \Big].
    \end{equation*}
    On the other hand, for each $\beta \in [0, 1]$ we have $r_\beta = r^A - \beta \cdot (r^H + r^A)$.

    Finally, for each $N \geq 1$, let $R^A_N = (\sum_{n=1}^N r^A_n) / N$ and $R^H_N = (\sum_{n=1}^N r^H_n) / N$. Then $0 \leq R^H_N, R^A_N \leq l$ for each $N \geq 1$, where $l$ is the maximal length of a private fork. This is because, at every time step up to $l$ new blocks can be accepted in the main chain at the depth greater than $d$. Hence, since the MDP $\MDP$ under every strategy gives rise to an ergodic Markov chain and since we showed that $R^A_N$ and $R^H_N$ are bounded, it follows from the Strong Law of Large Numbers for Markov chains~\cite[Theorem~1.10.2]{norris} that under every MDP strategy the limits $\lim_{N\rightarrow\infty} R^A_N$ and $\lim_{N\rightarrow\infty} R^H_N$ exist and almost-surely converge to some value. Moreover, we also have $\lim_{N \rightarrow \infty} (R^A_N + R^H_N) \geq \delta > 0$ almost-surely, where $\delta = \frac{1-p}{1-p+p\cdot d\cdot f}$. To see this, recall from Section~\ref{sec:methmodel} that at every time step the probability of an honest miner adding a new block to the main chain and thus either $r^H_n = 1$ or $r^A_n = 1$ with probability $\delta = \frac{1-p}{1-p+p\cdot d\cdot f}$. 
    
    We are now ready to prove the theorem. To prove the first part of the theorem claim, observe that for each $\beta \in [0,1]$ and for each strategy $\sigma$ we have
    \begin{equation}\label{eq:eq1}
    \begin{split}
        \MP_{\beta}(\sigma) &= \mathbb{E}^{\sigma}\Big[ \liminf_{N\rightarrow\infty} \frac{\sum_{n=1}^N r^A_n - \beta \cdot (r^H_n + r^A_n)}{N}\Big] \\
        &= \mathbb{E}^{\sigma}\Big[ \liminf_{N\rightarrow\infty} R^A_N - \beta \cdot (R^H_N + R^A_N)\Big] \\
        &= \mathbb{E}^{\sigma}\Big[ \lim_{N\rightarrow\infty} R^A_N - \beta \cdot (R^H_N + R^A_N)\Big],
    \end{split}
    \end{equation}
    where the last limit exists since $\lim_{N\rightarrow\infty} R^A_N$ and $\lim_{N\rightarrow\infty} R^A_N + R^H_N$ both exist and are almost-surely finite. Hence, $\MP_\beta(\sigma) = 0$ if and only if $\mathbb{E}^{\sigma}[ \lim_{N\rightarrow\infty} R^A_N - \beta \cdot (R^H_N + R^A_N)] = 0$. Since $\lim_{N \rightarrow \infty} (R^A_N + R^H_N) \geq \delta > 0$ almost-surely, we may divide the expression in the expectation by $R^A_N + R^H_N$ to get that $\MP_\beta(\sigma) = 0$ if and only if
    \[ \RER(\sigma) = \mathbb{E}^{\sigma}\Big[ \lim_{N\rightarrow\infty} \frac{R^A_N}{R^H_N + R^A_N} \Big] = \beta. \]
    Thus, we have $\MP^\ast_\beta = \sup_\sigma \MP_\beta(\sigma) = 0$ if and only if 
    \[ \beta = \sup_\sigma \mathbb{E}^{\sigma}\Big[ \lim_{N\rightarrow\infty}\frac{R^A_N}{R^A_N + R^H_N}\Big] = \RER^\ast, \]
    which concludes the proof of the first part of the theorem claim.
    
    To prove the second part of the theorem claim, let $\epsilon > 0$, $\beta \in [\RER^\ast-\epsilon,\RER^\ast]$ and suppose that $\sigma$ is a strategy such that $\MP_{\beta}(\sigma) = \MP^\ast_{\beta}$. We need to show $\RER(\sigma) \in [\RER^\ast-\epsilon,\RER^\ast]$.

    Since $\beta \leq \RER^\ast$ and since $\MP^\ast_{x}$ is a monotonically decreasing function in $x\in[0,1]$, we have that $\MP_{\beta}(\sigma) = \MP^\ast_{\beta} \geq \MP^{\ast}_{\RER^\ast} = 0$, by the first part of the theorem. Analogously as in eq.~\eqref{eq:eq1}, we get
    \begin{equation*}
        \MP_{\beta}(\sigma) 
        = \mathbb{E}^{\sigma}\Big[ \lim_{N\rightarrow\infty} R^A_N - \beta \cdot (R^H_N + R^A_N)\Big] \geq 0.
    \end{equation*}
    Then, since $\lim_{N \rightarrow \infty} (R^A_N + R^H_N) \geq \delta > 0$ almost-surely, we may divide the expression in the expectation by $R^A_N + R^H_N$ to get
    \begin{equation*}
        \mathbb{E}^{\sigma}\Big[ \lim_{N\rightarrow\infty} \frac{R^A_N}{R^A_N = R^H_N} - \beta\Big] \geq 0.
    \end{equation*}
    and thus $\RER(\sigma) \geq \beta$. Thus, as $\beta \in [\RER^\ast-\epsilon,\RER^\ast]$, we conclude that $\RER(\sigma) \in [\RER^\ast-\epsilon,\RER^\ast]$, as desired (where $\RER(\sigma) \leq \RER^\ast$ follows by the definition of $\RER^\ast)$.
\end{proof}

\end{document}